\documentclass[12pt]{article}

\usepackage[T1]{fontenc}
\usepackage{textcomp}
\usepackage[latin9]{inputenc}
\usepackage{color}
\usepackage{float}
\usepackage{amsmath}
\usepackage{amsthm}
\usepackage{amssymb}
\usepackage{graphicx}
\usepackage{geometry}
\usepackage{setspace}
\usepackage{esint}
\usepackage[authoryear]{natbib}
\usepackage[bookmarks=true,bookmarksnumbered=false,bookmarksopen=false,
 breaklinks=true,pdfborder={0 0 1},backref=section,colorlinks=true]
 {hyperref}
\hypersetup{
 linkcolor=blue,urlcolor=blue,citecolor=blue,setpagesize=false,linktocpage}

\makeatletter
\theoremstyle{plain}
\newtheorem{lem}{\protect\lemmaname}
\theoremstyle{plain}
\newtheorem{prop}{\protect\propositionname}
\theoremstyle{plain}
\newtheorem{assumption}{\protect\assumptionname}

\usepackage{color}
\usepackage{subcaption}
\usepackage{enumerate}
\usepackage{tikz}
\usepackage{pgf}
\usetikzlibrary{patterns}

\theoremstyle{definition}

\makeatother

\providecommand{\assumptionname}{Assumption}
\providecommand{\lemmaname}{Lemma}
\providecommand{\propositionname}{Proposition}

\begin{document}
\title{Heterogeneity, trade integration and spatial inequality\thanks{This work is dedicated to Carlos Hervés-Beloso. I am indebted to Sofia
B. S. D. Castro and João Correia da Silva for their invaluable comments
and suggestions and also for their help in revising previous and related
versions of this work. I am also very thankful to the audience at
the "Porto Workshop on Economic Theory and Applications - in honor
of Carlos Hervés-Beloso" for the fruitful discussions. This research
has been financed by Portuguese public funds through FCT - Fundação
para a Ciência e a Tecnologia, I.P., in the framework of the projects
with references UIDB/00731/2020 and PTDC/EGE-ECO/30080/2017 and by
Centro de Economia e Finanças (CEF.UP), which is financed through
FCT. Part of this research was developed while José M. Gaspar was
a researcher at the Research Centre in Management and Economics, Católica
Porto Business School, Universidade Católica Portuguesa, through the
grant CEECIND/02741/2017.}}
\author{José M. Gaspar\thanks{School of Economics and Management and CEF.UP, University of Porto.
Email: jgaspar@fep.up.pt.}}
\date{\vspace{-5ex}
}

\maketitle
 
\begin{abstract}
We study the impact of economic integration on agglomeration in a
model where all consumers are inter-regionally mobile and have heterogeneous
preferences regarding their residential location choices. This heterogeneity
is the unique dispersion force in the model. We show that, under reasonable
values for the elasticity of substitution among varieties of consumption
goods, a higher trade integration always promotes more symmetric spatial
patterns, reducing the spatial inequality between regions, irrespective
of the functional form of the dispersion force. We also show that
an increase in the degree of heterogeneity in preferences for location
leads to a more even spatial distribution of economic activities and
thus also reduces the spatial inequality between regions.
\end{abstract}
\bigskip{}



\section{Introduction}

In this paper, we study the possible spatial distributions in the
$2$-region Core-Periphery (CP) model by Murata (2003), where all
labour is free to migrate between regions and workers are heterogeneous
regarding their preferences for location. Our aim is to investigate
how spatial patterns evolve as regions become more integrated and
whether the degree of heterogeneity has qualitative impact on this
relationship. Heterogeneity in preferences for location generates
a local dispersion force (akin to congestion), which we model as a
utility penalty, following the recent work of Castro et al. (2022).
This framework can been shown to encompass various functional forms
used in discrete choice models of probabilistic migration, including
the well-known Logit model (used e.g. by Murata, 2003; Tabuchi and
Thisse, 2002).

We show that, irrespective of the degree of heterogeneity, a higher
trade integration always promotes a more symmetric spatial distribution
of economic activities, thus reducing the spatial inequality in terms
of industry size between regions. Since heterogeneity does not co-vary
with transportation costs, the latter only affect the strength of
agglomeration forces due to increasing returns. This is a reversion
of the usual prediction that lower transport costs lead to agglomeration
(Krugman, 1991; Fujita et al., 1999, Chap. 5), and contrasts the findings
that consumer heterogeneity leads to a bell-shaped relationship between
decreasing transport costs and the spatial distribution of economic
activities (Murata, 2003; Tabuchi and Thisse, 2002). Since there is
no immobile workforce (no fixed regional internal demand), there is
a lower incentive for firms to relocate to smaller markets in order
to capture a higher share of local demand and avoid competition when
transport costs are higher. Krugman and Elizondo (1996), Helpman (1998)
and Murata and Thisse (2005), have also obtained similar results concerning
the relationship between transportation costs and the spatial distribution
of economic activities. The first includes a congestion cost in the
core, the second considers a non-tradable housing sector and the treatment
is only numerical, while the latter's prediction cannot be disassociated
from the interplay between inter-regional transportation costs and
intra-regional commuting costs. Recently, Tabuchi et al. (2018) have
reached similar conclusions to ours arguing that falling transport
costs increase the incentives for firms in peripheral regions to increase
production since they have a better access to the core, thus contributing
to the dispersion of economic activities. Allen and Arkolakis (2014)
estimated that the removal of the US Interstate Highway System, which,
by limiting accesses, can be interpreted as an increase in transportation
costs, would cause a redistribution of the population from more economically
remote regions to less remote regions in the US. This adds empirical
validity to our findings.

The rest of the paper is organized as follows. In section 2 we describe
the model and the short-run equilibrium. In section 3 we study the
qualitative properties of the long-run equilibria. Finally, section
4 is left for some concluding remarks.

\section{The model}

There are two regions $L$ and $R$, symmetric in all respects, with
a total population of mass $1$. Consumers are assumed to have heterogeneous
preferences with respect to the region in which they reside. We follow
Castro et al. (2022) and assume that these preferences are described
by a parameter $x,$ uniformly distributed along the interval $[0,1]$.
The consumer with preference described by $x=0$ (resp. $x=1$) has
the highest preference for residing in region $L$ (resp. region $R$).
An agent whose preference corresponds to $x=1/2$ is indifferent between
either region. We can thus refer to each consumer with respect to
his preference towards living in region $L$ as $x\in[0,1]$.

For a consumer with preference $x$, the utilities from living in
region $L$ and $R$ are given, respectively, by: 
\begin{align}
U_{L}(x)= & U\left(u\left(C_{L}\right),t(x)\right)\nonumber \\
U_{R}(x)= & U\left(u\left(C_{R}\right),t(1-x)\right),\label{eq: utility}
\end{align}
where $C_{i}$ denotes the level of consumption of a consumer living
in region $i\in\left\{ L,R\right\} $, and $t:[0,1]\mapsto\mathbb{R}$
is the utility penalty of consumer $x$ associated with living in
$L$, while $t(1-x)$ is the utility penalty the same consumer faces
from living in region $R.$ We assume that $t(x)$ is differentiable,
such that $t^{\prime}(x)>0,\forall x\in[0,1]$. 

The number of agents in region $L$ is given by $h\in[0,1]$ and fully
describes the spatial distribution of economic activities. In the
short-run, the spatial distribution of workers and firms is fixed
and an equilibrium is reached when all prices, wages, and output adjust
to clear markets given $h$. In what follows, the short-run equilibrium
is equivalent to Murata (2003) and Tabuchi et al. (2018), which we
briefly describe here. The consumption aggregate is a CES composite
given by: 
\[
C=\left[\int_{s\in S}c(s)^{\tfrac{\sigma-1}{\sigma}}ds\right]^{\tfrac{\sigma}{\sigma-1}},
\]
where $s$ stands for the variety produced by each monopolistically
competitive firm and $\sigma$ is the constant elasticity of substitution
between any two varieties. The consumer is subject to the budget constraint
$P_{i}C_{i}=w_{i}$, where $P_{i}$ is the price index and $w_{i}$
is the nominal wage in region $i$. Utility maximization by a consumer
in region $i$ yields the following demand for each manufactured variety
produced in $j$ and consumed in $i$: 
\begin{equation}
c_{ij}=\dfrac{p_{ij}^{-\sigma}}{P_{i}^{1-\sigma}}w_{i},\label{eq:individual demand}
\end{equation}
where: 
\begin{equation}
P_{i}=\left[\int_{s\in S}p_{i}(s)^{1-\sigma}d(s)\right]^{\tfrac{1}{1-\sigma}},\label{eq:price index 1}
\end{equation}
is the manufacturing price index for region $i$. A manufacturing
firm faces the following cost function: 
\[
TC(q)=w\left(\beta q+\alpha\right),
\]
where $q$ corresponds to a firm's total production of manufacturing
goods, $\beta$ is the input requirement (per unit of output) and
$\alpha$ is the fixed input requirement. The manufacturing good is
subject to trade barriers in the form of iceberg costs, $\tau\in(1,+\infty)$:
a firm ships $\tau$ units of a good to a foreign region for each
unit that arrives there. Assuming free entry in the manufacturing
sector, at equilibrium firms will earn zero profits, which translates
into the following condition: 
\begin{align*}
\pi\equiv\left(p-\beta w\right)q-\alpha w & =0,
\end{align*}
which gives the firm's total equilibrium output, symmetric across
regions: 
\begin{equation}
q_{L}=q_{R}=\frac{\alpha(\sigma-1)}{\beta},\label{eq: total production}
\end{equation}
and the following profit maximizing prices: 
\begin{equation}
p_{LL}=\dfrac{\beta\sigma}{\sigma-1}w_{L}\text{ \ and\  }p_{RL}=\dfrac{\beta\sigma\tau}{\sigma-1}w_{L},\label{eq:optimal prices}
\end{equation}

\noindent where $p_{ij}$ is the price of a good produced in region
$i$ and consumed in region $j$. Labour-market clearing implies that
the number of agents in region $L$, equals labour employed by a firm
times the number of varieties (and hence firms) $n_{L}$ produced
in region $L$: 
\[
h=n_{L}\left(\alpha+\beta q_{L}\right),
\]
from where, using (\ref{eq: total production}), we get the number
of firms in each region $i$: 
\begin{equation}
n_{L}=\dfrac{h}{\sigma\alpha}\text{ \ and \ }n_{R}=\dfrac{1-h}{\sigma\alpha}.\label{eq:number of varieties}
\end{equation}

\noindent Choosing labour in region $R$ as the numeraire, we can
normalize $w_{R}$ to $1$. The price indices in $L$ and $R$ using
(\ref{eq:price index 1}) are given, after (\ref{eq:number of varieties}),
by:{\small{} 
\begin{align}
P_{L} & =\left[\dfrac{h}{\sigma\alpha}\left(\dfrac{\beta\sigma}{\sigma-1}w\right)^{1-\sigma}+\dfrac{1-h}{\sigma\alpha}\left(\dfrac{\beta\sigma\tau}{\sigma-1}\right)^{1-\sigma}\right]^{\tfrac{1}{1-\sigma}},\nonumber \\
P_{R} & \text{ =\ensuremath{\left[\dfrac{1-h}{\sigma\alpha}\left(\dfrac{\beta\sigma}{\sigma-1}\right)^{1-\sigma}+\dfrac{h}{\sigma\alpha}\left(\dfrac{\beta\sigma\tau}{\sigma-1}w\right)^{1-\sigma}\right]^{\tfrac{1}{1-\sigma}}}.}\label{eq: price index 2}
\end{align}
}{\small\par}

\noindent Rewriting a firm's profit in region $L$ as: 
\[
\pi_{L}=\left(p_{LL}-\beta w\right)\left[hc_{LL}+(1-h)\tau c_{RL}\right]-\alpha w,
\]
the zero profit condition yields, after using (\ref{eq:individual demand}),
(\ref{eq:optimal prices}) and (\ref{eq: price index 2}), the following
wage equation: 
\begin{equation}
h=\dfrac{w^{\sigma}-\phi}{w^{\sigma}-\phi+w(w^{-\sigma}-\phi)},\label{eq:short run equilibrium}
\end{equation}
where $\phi\equiv\tau^{1-\sigma}\in(0,1)$ is the \emph{freeness of
trade}.

The nominal wage $w$ can be implicitly defined as a function of the
spatial distribution in region $L$, $h\in[0,1]$.\footnote{The conditions for application of the Implicit Function Theorem are
shown to be satisfied in Appendix A.} We have $w(0)=\phi^{1/\sigma}<1$, $w(1/2)=1$ and $w(1)=\phi^{-1/\sigma}>1$.
We say that $L$ is larger than $R$ when $h>1/2,$ and smaller otherwise.

\begin{figure}[!h]
\begin{centering}
\includegraphics[scale=0.75]{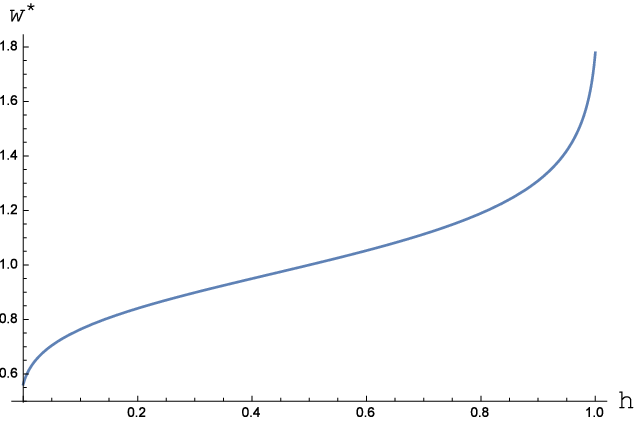} 
\par\end{centering}
\caption{Short-run equilibrium relative wage as a function of the consumers
in region $L$. We set $\sigma=2$ and $\phi\in\{0.1,0.5,0.7\}$.}
\label{fig 1} 
\end{figure}

Figure \ref{fig 1} illustrates $w(h)$ for $h\in[0,1]$ with parameter
values $\sigma=2$ and $\phi\in\{0.1,0.5,0.7\}$. The following result
corroborates this illustration and discusses the impact of trade barriers
on wages. 
\begin{lem}
\label{lem:The-relative-nominal}The relative nominal wage $w$ is
strictly increasing in $h$. When region $L$ is larger (smaller)
than region $R$, the wage $w$ is decreasing (increasing) in $\phi$. 
\end{lem}
\begin{proof}
See Appendix A. 
\end{proof}
We conclude that higher trade barriers increase the wage divergence
between the regions. The intuition is as follows. When all workers
are mobile, they can move to the region that offers them a relatively
better access to varieties. This advantage of the larger region becomes
higher as transport costs increase because markets become more focused
on local demand. 
This increases expenditures in the larger region relative to the smaller,
which in turn pushes the nominal relative wage upwards.

We consider a general isoelastic sub-utility for consumption goods:
\begin{equation}
\begin{cases}
u_{i}=\dfrac{C_{i}^{1-\theta}-1}{1-\theta}, & \text{ if }\theta\in[0,1)\cup(1,+\infty)\\
u_{i}=\ln(C_{i}), & \text{ if }\theta=1,
\end{cases},\label{eq:isoelastic utility}
\end{equation}
where for $\theta=1$ we take the limit value of the upper expression
of $u_{i}$. The parameter $\theta$ is a positive agglomeration externality.
It influences how a change in the regional consumption differential,
$C_{L}-C_{R}$, impacts the regional utility differential, $u_{L}-u_{R}$,
i.e., it influences the strength of the self-reinforcing agglomeration
mechanism when one region is more populated than the other. The specification
in (\ref{eq:isoelastic utility}) is quite general and encompasses
the Murata (2003) model as a particular case when $\theta=0$, which
shall be of great interest for comparison purposes in Section 4.\footnote{We note in advance that the qualitative results of the model obtained
in subsequent sections do not depend on $\theta$ (c.f., Section 3.2).}

The utility differential from consumption goods, $\Delta u\equiv u_{L}-u_{R},$
is given by: 
\begin{equation}
\Delta u=\begin{cases}
\dfrac{\eta}{1-\theta}\left\{ w^{1-\theta}\left[(1-h)\phi+hw^{1-\sigma}\right]^{\tfrac{1-\theta}{\sigma-1}}-\left[1-h+h\phi w^{1-\sigma}\right]^{\tfrac{1-\theta}{\sigma-1}}\right\} , & \ \text{if }\mbox{\ensuremath{\theta\neq1}}\\
\ln w+\dfrac{1}{\sigma-1}\ln\left[\frac{hw^{1-\sigma}+(1-h)\phi}{h\phi w^{1-\sigma}+(1-h)}\right], & \ \text{if }\theta=1.
\end{cases}\label{eq:utility differential-1}
\end{equation}
We adopt the normalizations by Fujita \emph{et al}. (1999), i.e.,
$\alpha\sigma=1$ so that the number of consumers in a region equals
its number of firms. Moreover, we assume that $(\sigma-1)/(\sigma\beta)=1$
so that the price of each manufactured variety in a region equals
its workers' nominal wage. These imply that $\eta=1$.\footnote{This choice is made mainly for convenience, as it allows us to abstract
from changes in the variable input requirement $\beta$.}

\begin{figure}[H]
\begin{centering}
\includegraphics[scale=0.7]{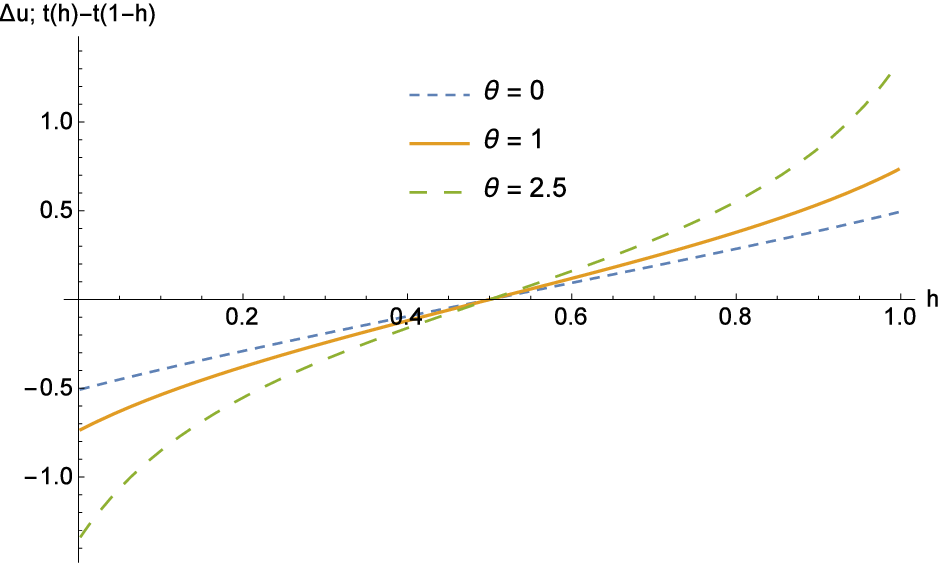}
\par\end{centering}
\caption{Utility differential $\Delta u=u_{L}-u_{R}$ for different levels
of $\theta$. Parameter values are $\sigma=2$ and $\phi=0.5$.}
\label{fig 2}
\end{figure}
Figure \ref{fig 2} depicts $\Delta u$ in (\ref{eq:utility differential-1})
for different levels of $\theta$. We observe that a higher $\theta$
increases the utility differential $\Delta u$ for any spatial distribution
$h>1/2$. Therefore, if $L$ is relatively more industrialized, a
higher $\theta$ increases the attractiveness of $L$ relative to
$R$ for consumers. Thus, it strengthens the agglomeration forces
towards region $L$. 

\section{Long-run equilibria}

In a long-run spatial equilibrium, each worker chooses to live in
the region that provides a higher utility. We follow Castro et al.
(2022) and assume that $t(x)$ enters additively in overall utility,
$U_{i}(x)$, so that the utility penalty, and hence heterogeneity
in preferences for location, are modelled \emph{à la }Hotelling (1929).\footnote{It could be interesting to consider the case where $t(x)$ is multiplicative
(e.g. as a percentage of $U_{i}(x)$). However, the non-linear impact
of $t(x)$ would likely make analytical results harder to obtain.} In a long-run equilibrium, workers with $x\in[0,h)$ live in region
$L$, and workers with $x\in(h,1]$ live in region $R$. Knowing from
utility maximization that $C_{L}=w/P_{L}$ and $C_{R}=1/P_{R}$, we
rewrite the indirect utilities of the consumer with $x=h$ as: 
\begin{align}
V_{L}(h) & =\dfrac{C_{L}^{1-\theta}-1}{1-\theta}-t(h)\nonumber \\
V_{R}(h) & =\dfrac{C_{R}^{1-\theta}-1}{1-\theta}-t(1-h),\label{eq: indirect utility}
\end{align}
where $h$ satisfies the short-run equilibrium in (\ref{eq:short run equilibrium}).

\subsection{Interior equilibria}

Due to symmetry, we focus only on the case where $L$ is larger or
the same size as $R$, i.e., $h\geq1/2$. We define an interior equilibrium
$h^{*}\in[1/2,1)$ as a spatial distribution $h$ that satisfies both
(\ref{eq:short run equilibrium}) and $V_{L}=V_{R}$. Such an equilibrium
is said to be stable if $d\left(V_{L}-V_{R}\right)/dh<0$. This leads
to the following Lemma. 
\begin{lem}
\label{prop:stability of partial agglomeration}An interior equilibrium
$h^{*}\in[1/2,1)$ is stable if:{\footnotesize{} 
\[
\zeta\left\{ \left(\varphi w^{\sigma\frac{1-\theta}{\sigma-1}}+\frac{\psi}{w}\right)\left\{ \frac{\left(1-\phi^{2}\right)w}{w^{2\sigma}-\left[\varphi(w+1)w^{\sigma}\right]+w}\right\} ^{\frac{1-\theta}{\sigma-1}}\right\} <t^{\prime}\left(h^{*}\right)-t^{\prime}\left(1-h^{*}\right),
\]
}where $h^{*}$ satisfies the short-run equilibrium condition in (\ref{eq:short run equilibrium}),
and:{\small{} 
\begin{align*}
\zeta & =\dfrac{w^{2\sigma}-\left[\phi(w+1)w^{\sigma}\right]+w}{(\sigma-1)\left[(\sigma-1)\phi+(\sigma-1)\phi w^{2\sigma}+\left(-2\sigma+\phi^{2}+1\right)w^{\sigma}\right]};\\
\varphi & =w^{-\theta-\sigma}\left\{ \phi\left[\sigma+(\sigma-1)w\right]w^{\sigma}-2\sigma w+w\right\} ;\\
\psi & =(\sigma-1)\phi+(1-2\sigma)w^{\sigma}+\sigma\phi w.
\end{align*}
}{\small\par}

\end{lem}
\begin{proof}
See Appendix B. 
\end{proof}
\noindent An interior equilibrium is called \emph{symmetric dispersion
}if $h^{*}=1/2$ and \emph{partial agglomeration }otherwise. While
the former always exists, the latter depends on the form of $t(x)$.
Knowing that $w=1$ for $h=1/2$ the stability condition in Lemma
\ref{prop:stability of partial agglomeration} for symmetric dispersion
simplifies to:

\begin{equation}
\dfrac{d\Delta u}{dh}\left(\tfrac{1}{2}\right)\equiv\frac{2(2\sigma-1)(1-\phi)\left(\frac{1+\phi}{2}\right)^{\frac{1-\theta}{\sigma-1}}}{(\sigma-1)(2\sigma+\phi-1)}<t^{\prime}\left(\tfrac{1}{2}\right).\label{eq:stability symmetric dispersion}
\end{equation}

\noindent As expected, symmetric dispersion is stable if, after a
small increase in $h$, the relative utility gain from consumption
is lower than the increase in the utility penalty for the agent of
type $x=h$. A careful inspection of (\ref{eq:stability symmetric dispersion})
allows us to conclude that the LHS is decreasing in $\phi$ if $\theta\geq1$.
For $\theta<1$, it is decreasing in $\phi$ if $\sigma\geq1+\sqrt{2}/2\approx1.71$.
Symmetric dispersion then becomes easier to sustain under lower transportation
costs if $\sigma>1+\sqrt{2}/2$, which, according to recent empirical
estimations for $\sigma$, is more than reasonable.\footnote{Estimations evidence that $\sigma$ should be significantly larger
than unity (Crozet, 2004; Head and Mayer, 2004; Niebuhr, 2006; Bosker
\emph{et al., }2010). Anderson and Wincoop (2004), for instance, find
that it is likely to range from 5 to 10.} As $\phi\rightarrow1$, symmetric dispersion is always stable because
$t^{\prime}\left(\tfrac{1}{2}\right)>0$. This warrants the following
result.
\begin{prop}
\textup{\label{prop:As--increases}As $\phi$ increases from a low
level, symmetric dispersion either remains stable or becomes stable,
if either conditions hold: }\textup{\emph{(i).}}\textup{ $\theta\in[0,1)$
and $\sigma>1+\sqrt{2}/2$; or }\textup{\emph{(ii).}}\textup{ $\theta\geq1$. }
\end{prop}
Thus, an increase in $\phi$ can never turn symmetric dispersion from
stable to unstable, suggesting that lower trade barriers always promote
an even distribution of economic activities. Since, at symmetric dispersion,
consumers in each region have access to the same amount of manufactures,
an exogenous migration will induce a lower (higher) benefit from local
consumption goods in the larger market if transport costs are lower
(higher). This is captured by the fact that the relative decrease
in prices and increase in wages (at the symmetric equilibrium) is
more pronounced when transport costs are higher.

Let $\phi=\phi_{b}\in(0,1)$ such that $\frac{d\Delta u}{dh}\left(\tfrac{1}{2};\phi_{b}\right)=t^{\prime}\left(\tfrac{1}{2};\phi_{b}\right)$.
If $\phi_{b}$ exists, then symmetric dispersion is stable for $\phi>\phi_{b}$
and unstable for $\phi<\phi_{b}$. We have the following result.
\begin{lem}
\label{lem:A-curve-of}A curve of partial agglomeration equilibria
$h^{*}\in(1/2,1)$ branches from $\phi=\phi_{b}.$ If it branches
in the direction of $\phi<\phi_{b}$, partial agglomeration is stable.
\end{lem}
\begin{proof}
See Appendix B.
\end{proof}
The result above establishes existence of partial agglomeration in
a neighbourhood of $\phi=\phi_{b}$.\footnote{Uniqueness is not ensured for the entire range of $\phi\in(0,\phi_{b})$
as the model may undergo secondary bifurcations along the primary
branch of partial agglomeration equilbria, such as a saddle-node bifurcation
(see e.g. Gaspar et al. (2018, 2021) and Castro et al. (2021) for
applications in economic geography).} Let us now assume that a partial agglomeration $h^{*}\in(1/2,1)$
equilibrium exists and is stable for some value of $\phi\in(0,1)$.
Regarding the influence of trade integration on partial agglomeration
equilibria, we establish the following result. 
\begin{prop}
\label{prop:As-trade-barriers decrease}If a stable partial agglomeration
$h^{*}\in(1/2,1)$ exists, it becomes more symmetric as trade barriers
decrease. 
\end{prop}
\begin{proof}
See Appendix B. 
\end{proof}
This result states that if most consumers reside in region $L$, increasing
the freeness of trade will lead to a smooth exodus from region $L$
to region $R$, irrespective of the value of $t(h^{*})$. We conclude
from Propositions \ref{prop:As--increases} and \ref{prop:As-trade-barriers decrease}
that more integration leads agents to distribute more equally among
the two regions, as in Helpman (1998), who also considers that all
labour is inter-regionally mobile. This qualitative effect of $\phi$
does \emph{not} depend on consumer heterogeneity. It thus contrasts
the findings that consumer heterogeneity leads to a bell-shaped relationship
between decreasing transport costs and agglomeration (Murata, 2003;
Tabuchi and Thisse, 2002).

\subsection{Logit preferences for location}

We begin this Section by noting that, for any value of $\theta$,
there exists a non-trivial subset of a suitably defined parameter
space over which the qualitative results of the model remain invariant.
Therefore, for the remainder of this Section, we will assume that
utility is logarithmic in consumption. This greatly simplifies the
analysis and entails no loss of generality.
\begin{assumption}
\label{assu:Let-.}Utility $u_{i}$ is logarithmic in $C_{i}$: $\theta=1\implies u_{i}=\ln(C_{i})$.
\end{assumption}
For the two region case, and following Tabuchi and Thisse (2002) and
Murata (2003), we assume that the probability that a consumer will
choose to reside in region $L$ is given by the Logit model, written
as: 
\begin{equation}
\Phi_{L}(h)=\dfrac{e^{u_{L}(h)/\mu}}{e^{u_{L}(h)/\mu}+e^{u_{R}(h)/\mu}},\label{eq:binary logit expression}
\end{equation}
where $\mu\geq0$ is a scale parameter which measures the dispersion
of consumer preferences. If $\mu=0$, consumers do not care about
their location preferences but rather solely about relative wages.
The number of agents $h$ is the same as the consumer $x$ who is
indifferent between living in region $L$ or in region $R$. Therefore,
the probability $\Phi_{L}(h)$ is tantamount to the indifferent consumer
$x=h$ (as in Castro et al., 2022). Thus, using (\ref{eq:binary logit expression}),
we can write: 
\begin{equation}
h=\dfrac{e^{u_{L}(h)/\mu}}{e^{u_{L}(h)/\mu}+e^{u_{R}(h)/\mu}}.\label{eq:indifferent logit}
\end{equation}
Manipulating (\ref{eq:indifferent logit}) yields: 
\begin{equation}
u_{L}+\mu\ln(1-h)=u_{R}+\mu\ln(h),\label{eq:utility equalization under logit}
\end{equation}
which rearranged yields the long-run equilibrium condition, $\Delta u(h)=\Delta t(h),$
with $t(x)=-\mu\ln(1-x).$ This yields $\Delta t(h)=\mu\left[\ln h-\ln(1-h)\right]$,
dubbed by Castro et al. (2022) as the ``home-sweet-home'' effect,
which comprises the model's unique dispersion force generated by consumer
heterogeneity in preferences for location.

Notice from (\ref{eq:utility equalization under logit}) that $t(h)=-\mu\ln(1-h)>0$
for $h\in[0,1]$. The overall utility of a consumer in region $i$
thus lies on the interval $(-\infty,u_{i}]$. For a strictly positive
$\mu$, the consumer who likes region $R$ the most (at $x=1$) will
never want to live in $L$, because, at $h=1$, the overall utility
loss from living in region $L$ is unbounded. This means that the
Logit model penalizes (benefits) the consumers who are less (more)
willing to leave a region very strongly. Therefore, agglomeration
in any region ($h\in\{0,1\})$ is not possible. Now, suppose condition
(\ref{eq:stability symmetric dispersion}) does not hold and symmetric
dispersion is unstable. Then, we have $\frac{d\Delta u}{dh}\left(\tfrac{1}{2}\right)>t^{\prime}\left(\tfrac{1}{2}\right)$
which means that $\Delta u(h)>t(h)-t(h)$ for $h=\tfrac{1}{2}+\epsilon,$
with $\epsilon>0$ small enough. Since $\Delta u\left(1\right)-\Delta t(1)\rightarrow-\infty$,
then, by the Intermediate Value Theorem, we must have $\Delta V=0$
for some $h\in(1/2,1)$. Accordingly, a partial agglomeration equilibrium
$h^{*}\in(1/2,1)$ must always exist when symmetric dispersion is
unstable.

The discussion above allows us to formalize our next result, which
summarizes the possible spatial outcomes under Logit type preferences,
depending on the degree of heterogeneity $\mu.$
\begin{prop}
\label{prop:Partial agglomeration logit model}Under Logit type preferences
and Assumption \ref{assu:Let-.}, the spatial distribution depends
on the level of consumer heterogeneity as follows: 
\end{prop}
\begin{itemize}
\item \emph{Symmetric dispersion is the only stable equilibrium if: 
\begin{equation}
\mu>\mu_{d}\equiv\frac{(2\sigma-1)(1-\phi)}{(\sigma-1)(2\sigma+\phi-1)};\label{eq:stability of partial agglomeration in the logit model}
\end{equation}
} 
\item \emph{Partial agglomeration $h^{*}\in\left(\tfrac{1}{2},1\right)$
is the only stable equilibrium and is unique if $\mu<\mu_{d}$.} 
\end{itemize}
\begin{proof}
See Appendix C. 
\end{proof}
When consumer heterogeneity is low there is a single stable partial
agglomeration equilibrium that is very asymmetric (close to agglomeration).
As consumer heterogeneity increases, partial agglomeration corresponds
to a more even distribution. Finally, if consumer heterogeneity is
high enough ($\mu>\mu_{d}$ in (\ref{eq:stability of partial agglomeration in the logit model})),
consumers disperse symmetrically across the two regions.

\begin{figure}[H]
\begin{centering}
\includegraphics[scale=0.7]{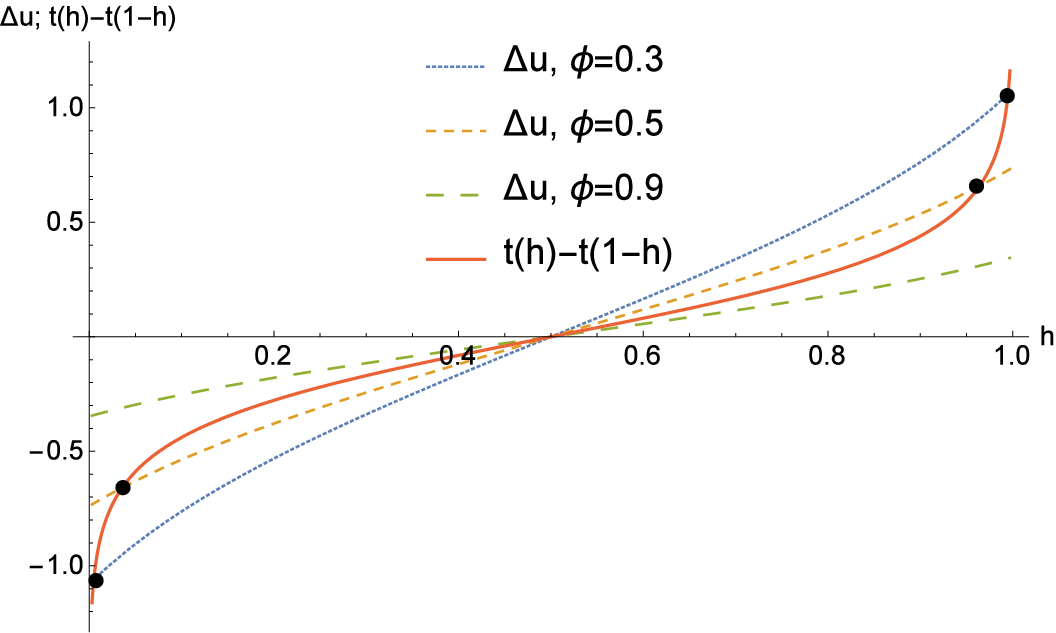}
\par\end{centering}
\caption{{\small{}Utility differential $\Delta u$ (dashed lines) and penalty
differential $t(h)-t(1-h)$ (thick line) for increasing levels of
$\phi$. For $\phi=0.3$ (upper dashed line), dispersion is unstable
and partial agglomeration is stable. For $\phi=0.5$ (medium dashed
line) dispersion is unstable, partial agglomeration is stable and
is less asymmetric. For $\phi=0.9$ (lower dashed line) dispersion
is stable and is the only equilibrium. The values for the parameters
are $\sigma=2.5$, $\theta=0$ and $\mu=0.2$.}}
\label{fig 5}
\end{figure}

In Figure \ref{fig 5} we show how the spatial distribution of industry
changes as regions become more integrated. We set $\sigma=2.5$ and
$\mu=0.2$. For low levels of the freeness of trade, a unique stable
partial agglomeration equilibrium exists where most consumers reside
in region $L$. This means that there is just a small fraction of
consumers in $R$ that are not willing to forego their preferred region
because the gain in consumption goods from doing so is not high enough.
As regions become more integrated, the home-market effect becomes
weaker and partial agglomeration becomes more symmetric. Finally,
for a high level of inter-regional integration, agglomeration forces
are so weak that no consumer is willing to leave his most preferred
region and only symmetric dispersion is stable.

\begin{figure}[H]
\begin{centering}
\includegraphics[scale=0.8]{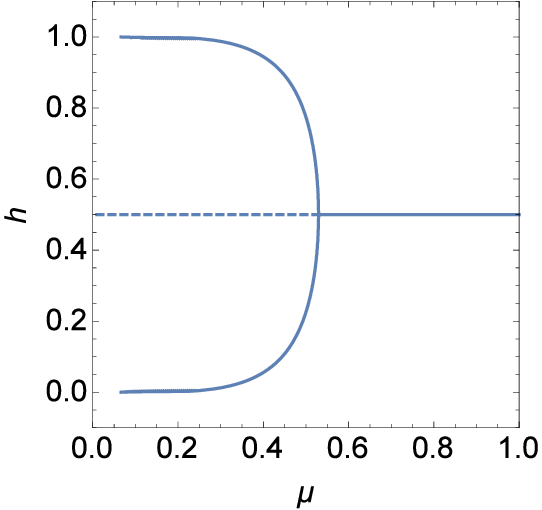} \includegraphics[scale=0.8]{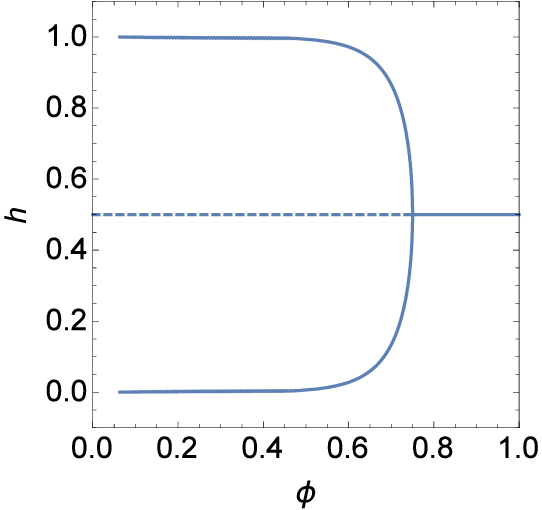}
\par\end{centering}
\caption{{\small{}Bifurcation diagrams. To the left, the bifurcation parameter
is $\mu\in[0,1]$ with $\phi=0.4$. To the right, we have $\phi\in(0,1)$
as the bifurcation parameter and set $\mu=0.2$. For both scenarios
we use $\sigma=2$ and $\theta=0$.}}
\label{fig 6}
\end{figure}

The preceding analysis is summarized by the two bifurcation diagrams
in Figure \ref{fig 6}, along a smooth parameter path where $\mu$
(left picture) and $\phi$ (right picture) increase. We uncover two
pitchfork bifurcations in $\mu$ and $\phi$. We confirm analytically
that they are supercritical for the case of logit preferences and
under Assumption \ref{assu:Let-.}.\footnote{Employing $\phi$ as the bifurcation parameter, we have $\tfrac{\partial^{3}\Delta V}{\partial h^{3}}(\tfrac{1}{2};\phi_{b})<0$,
with $\phi_{b}=\frac{(2\sigma-1)\left[\mu(\sigma-1)-1\right]}{-(\mu+2)\sigma+\mu+1}$
and $\sigma\in\left(2,\frac{\mu+1}{\mu}\right)$. Then by the proof
of Lemma \ref{lem:A-curve-of}\textbf{ }(Appendix B) the pitchfork
is supercritical. For the case of $\mu$ as the bifurcation parameter,
see the proof of Proposition \ref{prop:Partial agglomeration logit model}
(Appendix C).} This implies that a curve of stable partial agglomeration equilibria
branches in the direction of $\mu<\mu_{d}$ (left picture) and $\phi<\phi_{b}$
(right picture), showing that lower trade barriers or a higher heterogeneity
promote more even spatial distributions. The picture to the right
differs from other supercritical pitchforks found in the literature
(e.g., Pflüger, 2004; Ikeda et al., 2022) crucially in the sense that
the direction of change in stability as $\phi$ increases is reversed,
i.e., lower trade barriers leads to more symmetric spatial distributions,
as expected from our analysis. Regarding the effect of $\mu$, it
is quite intuitive that a higher scale of heterogeneity increases
local dispersion forces (i.e., dispersion forces that do not depend
on $\phi$) and thus also promotes a more even distribution of economic
activities.

\section{Discussion and concluding remarks}

We have seen that heterogeneity in preferences for location alone
bears no impact on the relationship between trade integration and
agglomeration, which is a monotonic decreasing one. This contrasts
the findings in other works with heterogeneity in consumer preferences,
such as Tabuchi and Thisse (2002) and Murata (2003), who show evidence
of a bell-shaped relationship between trade integration and the spatial
distribution of industry, whereby agglomeration but is followed by
a re-dispersion phase as trade integration increases. The former's
setting differs from ours because the authors consider an inter-regionally
immobile workforce whose role as a dispersive force is enhanced by
higher transportation costs.

However, it is particularly worthwhile to discuss the results of Murata
(2003), who uses a model that is a particular case of ours if we set
$\theta=0$ and consider logit type heterogeneity. Murata (2003) found
that the relationship between trade integration and agglomeration
need not be monotonic and depends on the degree of consumer heterogeneity,
which is at odds with our findings. For instance, for an intermediate
degree of consumer heterogeneity $(\mu)$, Murata finds that increasing
trade integration initially fosters agglomeration and later leads
to re-dispersion of industry. However, these conclusions can be shown
to stem from the author's particular choice of the value for the elasticity
of substitution, $\sigma=1.25$. As we have argued in Section 3, such
a low value is empirically implausible. For exceedingly low values
$(\sigma<1.71)$, increasing returns at the firm level are too strong.
Strong enough that the utility gain at dispersion becomes increasing
in $\phi$, instead of decreasing. This would justify an initial concentration
of industry as a result of an increase in $\phi.$ However, we have
seen that for a plausible range of $\sigma$ the utility gain from
consumption at symmetric dispersion always decreases with $\phi$.
Moreover, if $\theta>0$, this holds even for lower values of $\sigma$.\footnote{If $\theta\geq1$, the result holds for $\sigma>1$.}
Therefore, a higher $\phi$ always promotes more equitable distributions
as opposed to asymmetric ones.

Hence, when workers are completely mobile, more trade integration
ubiquitously leads to a more even dispersion of spatial distributions
among the two regions, irrespective of the degree of heterogeneity
in location preferences. Therefore, a \emph{de facto} lower inter-regional
labour mobility induced by consumer heterogeneity alone cannot account
for the predictions that a higher inter-regional integration will
lead to more spatial inequality or an otherwise bell-shaped relationship
between the two.

By considering that all consumers are allowed to migrate if they so
desire, we have shown that a higher inter-regional trade integration
always leads to more dispersed spatial distributions. This result
is independent of the level and impact of consumer heterogeneity.
This conclusion may be of potential use for policy makers. Namely,
the predictions that globalization is likely to lead to higher spatial
inequality in the future (World Bank, 2009) may be reversed if policies
are undertaken to promote inter-regional mobility.

\appendix

\section*{Appendix A - Wages and freeness of trade}

\textbf{Proof of Lemma } \ref{lem:The-relative-nominal}:

\noindent Consider $h=f(w)$ in (\ref{eq:short run equilibrium}).
Let us now define $F(h,w)=f(w)-h=0.$ Differentiation of $F(h,w)$
yields: 
\begin{equation}
\dfrac{\partial F(h,w)}{\partial w}=\dfrac{w^{\sigma}G(w^{\sigma})}{\left[w^{2\sigma}-(w+1)\phi w^{\sigma}+w\right]^{2}},\label{eq:derivative implicit function}
\end{equation}
where: 
\[
G(w^{\sigma})=-\left[\phi(\sigma-1)+(\sigma-1)\phi w^{2\sigma}-\left(2\sigma-\phi^{2}-1\right)w^{\sigma}\right].
\]
The derivative in (\ref{eq:derivative implicit function}) is zero
if $G(w^{\sigma})=0$. One can observe that $G(w^{\sigma})$ has either
two (real) zeros given by some $\left\{ w^{-},w^{+}\right\} ,$ or
none. Moreover, $G(w^{\sigma})$ is concave in $w^{\sigma}.$ Since
$G\left(w^{\sigma}=\phi\right)=\sigma\left[\phi(1-\phi^{2})\right]>0$
and $G\left(w^{\sigma}=\phi^{-1}\right)=-\left[\sigma\phi\left(1-\phi^{-2}\right)\right]>0$,
it must be that $G(w^{\sigma})>0$ for $w^{\sigma}\in[\phi,\phi^{-1}]\subset(w^{-},w^{+})$.

We now proceed to show that $w\in\left(0,\phi^{1/\sigma}\right)$
and $w\in\left(\phi^{-1/\sigma},+\infty\right)$ are not defined in
$h\in[0,1]$. Using (9), we have the following: 
\[
\lim_{w\rightarrow0}h(w)=0;\quad\lim_{w\rightarrow+\infty}h(w)=1;\quad h(\phi^{1/\sigma})=0;\quad h(\phi^{-1/\sigma})=1;
\]
Since $dh/dw=\partial F/\partial w$, we have that $h(w)$ is increasing
for $w^{\sigma}\in(w^{-},w^{+})$. The limits above, together with
the knowledge that the zeros of $dh/dw$ lie to the left and right
of $[\phi,\phi^{-1}]$, ensure that $h(w)<0$ for $w\in\left(0,\phi^{1/\sigma}\right)$
and $h(w)>1$ for $w\in\left(\phi^{-1/\sigma},+\infty\right)$.

Knowing that $\partial F/\partial h=-1$ and $\partial F/\partial w$
are continuous, by the IFT we can write $w:h\in[0,1]\subset\mathbb{R}\mapsto\mathbb{R}$
such that $dw/dh$ exists and $F(h,w(h))=0$.

Using implicit differentiation on (\ref{eq:short run equilibrium}),
we get: 
\begin{equation}
\dfrac{dw}{dh}=\frac{\left[w^{2\sigma}-(w+1)\phi w^{\sigma}+w\right]^{2}}{w^{\sigma}G(w^{\sigma})}.\label{eq:derivative of nominal wage}
\end{equation}
This derivative is positive for $w\in[\phi^{1/\sigma},\phi^{-1/\sigma}]$,
which implies that $w(h)$ is increasing in $[0,1]$. This concludes
the proof for the first statement.

From $F(h,w)=0$ defined above, differentiating with respect to $\phi,$
we get: 
\begin{align*}
\dfrac{\partial F}{\partial\phi}+\dfrac{\partial F}{\partial w}\dfrac{dw}{d\phi} & =0\iff\\
\dfrac{dw}{d\phi} & =-\left.\dfrac{\partial F}{\partial\phi}\right/\dfrac{\partial F}{\partial w}.
\end{align*}
Using (\ref{eq:short run equilibrium}) we get: 
\[
\dfrac{\partial F}{\partial\phi}=\frac{w^{\sigma+1}\left(w^{2\sigma}-1\right)}{\left[w^{2\sigma}-(w+1)\phi w^{\sigma}+w\right]^{2}},
\]
which is positive if $h>1/2$, because the latter implies $w>1$ (because
$w$ is increasing in $h$). As a result, we have $dw/d\phi<0$. Therefore,
the nominal wage $w$ is decreasing in the freeness of trade when
$L$ is the largest region ($h>1/2$). Conversely, we have $dw/d\phi>0$
if $h>1/2$. This concludes the proof of the second statement.\hfill{}$\square$ 

\section*{Appendix B - Interior Equilibria}

\textbf{Proof of Lemma \ref{prop:stability of partial agglomeration}:}

\noindent Taking the derivative of (\ref{eq:utility differential-1})
with respect to $h$, substituting for $w^{\prime}$ using (\ref{eq:derivative of nominal wage})
from Appendix A, and evaluating at $h^{*}$ using the short-run equilibrium
condition given by (\ref{eq:short run equilibrium}), we reach:{\footnotesize
\begin{equation}
\left.\dfrac{d\Delta u}{dh}\right|_{h=h^{*}}=\zeta\left\{ \left(\varphi w^{\sigma\frac{1-\theta}{\sigma-1}}+\frac{\psi}{w}\right)\left\{ \frac{\left(1-\phi^{2}\right)w}{w^{2\sigma}-\left[\varphi(w+1)w^{\sigma}\right]+w}\right\} ^{\frac{1-\theta}{\sigma-1}}\right\} ,\label{eq:stability of partial agglomeration}
\end{equation}
}with:{\small{} 
\begin{align*}
\zeta & =\dfrac{w^{2\sigma}-\left[\phi(w+1)w^{\sigma}\right]+w}{(\sigma-1)\left[(\sigma-1)\phi+(\sigma-1)\phi w^{2\sigma}+\left(-2\sigma+\phi^{2}+1\right)w^{\sigma}\right]};\\
\varphi & =w^{-\theta-\sigma}\left\{ \phi\left[\sigma+(\sigma-1)w\right]w^{\sigma}-2\sigma w+w\right\} ;\\
\psi & =(\sigma-1)\phi+(1-2\sigma)w^{\sigma}+\sigma\phi w.
\end{align*}
}The fact that an interior equilibrium is stable if $d\Delta u/dh<d\left[t(h)-t(1-h)\right]/dh$
at $h=h^{*}$ concludes the proof.\hfill{}$\square$ 

\bigskip{}

\noindent\textbf{Proof of Lemma \ref{lem:A-curve-of}:}

\noindent Suppose there exists a value of $\phi=\phi_{b}\in(0,1)$
such that $\dfrac{d\Delta u}{dh}\left(\tfrac{1}{2};\phi_{b}\right)=t^{\prime}\left(\tfrac{1}{2};\phi_{b}\right)$.
Let $\Delta V=V_{L}-V_{R}$. Then it is possible to show the following:
\begin{align*}
\frac{\partial\Delta V}{\partial h}\left(\tfrac{1}{2};\phi_{b}\right) & =0;\ \ \frac{\partial^{2}\Delta V}{\partial h^{2}}\left(\tfrac{1}{2};\phi_{b}\right)=0;\ \ \frac{\partial\Delta V}{\partial\phi}\left(\tfrac{1}{2};\phi_{b}\right)=0;\\
\frac{\partial^{2}\Delta V}{\partial h\partial\phi}\left(\tfrac{1}{2};\phi_{b}\right) & =-\frac{4\eta\left(\frac{1+\phi_{b}}{2}\right)^{\frac{1-\theta}{\sigma-1}}\left[\theta(1-\phi_{b})+2\sigma+\phi_{b}-3\right]}{(\sigma-1)^{2}(\phi_{b}+1)^{2}};
\end{align*}
The last expression is negative if we assume that $\sigma>1+\sqrt{2}/2$.
Next, we have:
\begin{align*}
\frac{\partial^{3}\Delta V}{\partial h^{3}}\left(\tfrac{1}{2};\phi_{b}\right)=\frac{16(\sigma-2)(2\sigma-3)(1-\phi_{b})^{3}\left(\frac{1+\phi_{b}}{2}\right)^{\frac{1}{\sigma-1}}}{(\sigma-1)^{3}(\phi_{b}+1)^{3}} & -2t^{\prime\prime\prime}\left(\frac{1}{2}\right).
\end{align*}
If $\tfrac{\partial^{3}\Delta V}{\partial h^{3}}\neq0$, then, according
to Guckenheimer and Holmes (2002), the model undergoes a pitchfork
bifurcation at symmetric dispersion and a curve of partial agglomeration
equilibria branches at $\phi=\phi_{b}.$ If the pitchfork is supercritical
($\tfrac{\partial^{3}\Delta V}{\partial h^{3}}<0$), the curve branches
in the direction of $\phi<\phi_{b}$ and partial agglomeration is
stable. If it is subcritical ($\tfrac{\partial^{3}\Delta V}{\partial h^{3}}>0)$,
it branches in the direction of $\phi>\phi_{b}$ and partial agglomeration
is unstable. The criticality of the bifurcation depends on the specific
functional form of $t(x)$. This concludes the proof.\hfill{}$\square$ 

\bigskip{}

\noindent\textbf{Proof of Proposition \ref{prop:As-trade-barriers decrease}:}

\noindent Assume that partial agglomeration $h^{*}=h\in\left(\frac{1}{2},1\right)$
exists and is stable. Then, according to Castro et al. (2022, Prop.
9, pp. 197), it becomes more symmetric if $\frac{d\Delta u}{d\phi}<0$.
We have: 
\[
\dfrac{d\Delta u}{d\phi}=\frac{\partial\Delta u}{\partial\phi}+\frac{\partial\Delta u}{\partial w}\frac{dw}{d\phi},
\]
which, by using (\ref{eq:short run equilibrium}) and (\ref{eq:utility differential-1}),
becomes: 
\begin{equation}
-\frac{w}{a_{3}}\left\{ a_{1}\left[\frac{\left(1-\phi^{2}\right)w^{\sigma+1}}{w^{2\sigma}-(w+1)\phi w^{\sigma}+w}\right]^{-\frac{\theta+\sigma-2}{\sigma-1}}+a_{2}\left[\frac{w(1-\phi^{2})}{w^{2\sigma}-(w+1)\phi w^{\sigma}+w}\right]^{-\frac{\theta+\sigma-2}{\sigma-1}}\right\} \label{eq:derivative}
\end{equation}
where: 
\begin{align*}
a_{1} & =w^{1-\theta}\left(\phi w^{\sigma}-1\right)\left[-2\sigma+2(\sigma-1)\phi w^{\sigma}+\phi^{2}+1\right]\\
a_{2} & =w^{-\sigma}\left(w^{\sigma}-\phi\right)\left[2(\sigma-1)\phi+\left(-2\sigma+\phi^{2}+1\right)w^{\sigma}\right]\\
a_{3} & =(\sigma-1)\left[w^{1-\sigma}+w^{\sigma}-(w+1)\phi\right]\left[(\sigma-1)\phi+(\sigma-1)\phi w^{2\sigma}+\left(-2\sigma+\phi^{2}+1\right)w^{\sigma}\right].
\end{align*}
Notice that (\ref{eq:derivative}) is equivalent to 
\[
-\frac{w\left[\frac{w(1-\phi^{2})}{w^{2\sigma}-(w+1)\phi w^{\sigma}+w}\right]^{-\frac{\theta+\sigma-2}{\sigma-1}}\left(a_{1}w^{-\sigma\frac{\theta+\sigma-2}{\sigma-1}}+a_{2}\right)}{a_{3}}.
\]
Thus, the sign of (\ref{eq:derivative}) is given by the sign of:
\[
\Psi\equiv-\frac{a_{1}w^{-\sigma\frac{\theta+\sigma-2}{\sigma-1}}+a_{2}}{a_{3}}.
\]
Since $w>1$ for $h\in(\frac{1}{2},1]$, we have $a_{3}<0$ because
$\left(w^{1-\sigma}+w^{\sigma}-(w+1)\phi\right)>0$ and 
\[
\left[(\sigma-1)\phi+(\sigma-1)\phi w^{2\sigma}+\left(-2\sigma+\phi^{2}+1\right)w^{\sigma}\right]<0.
\]
Next, taking the limit of $\Psi$ as $\theta\rightarrow+\infty$,
we get:
\[
\lim_{\theta\rightarrow+\infty}\Psi=w^{-\sigma}\left(w^{\sigma}-\phi\right)\left[2(\sigma-1)\phi+\left(-2\sigma+\phi^{2}+1\right)w^{\sigma}\right]<0.
\]

\noindent Differentiating $\Psi$ with respect to $\theta$ yields:
\[
\frac{\ln(w)w^{-\frac{\sigma(\theta+\sigma-2)}{\sigma-1}-\theta+1}\left(\phi w^{\sigma}-1\right)\left[-2\sigma+2(\sigma-1)\phi w^{\sigma}+\phi^{2}+1\right]}{\sigma-1},
\]
which is positive because $1<w^{\sigma}<1/\phi$ and $\left[-2\sigma+2(\sigma-1)\phi w^{\sigma}+\phi^{2}+1\right]<0$.
Thus, $\Psi$ remains negative for all $\theta\geq0$. We conclude
that (\ref{eq:derivative}) is negative and thus $\frac{d\Delta u}{d\phi}<0$
for any $\theta\in\left[0,1\right)\cup\left(1,+\infty\right)$. For
$\theta=1$, take $u_{i}=\ln(C_{i})$ to verify that the derivative
is negative. We conclude that $\frac{d\Delta u}{d\phi}<0,\forall\theta\geq0$.
This finishes the proof. \hfill{}$\square$

\pagebreak{}

\section*{Appendix C - Logit heterogeneity}

\noindent\textbf{Proof of Proposition \ref{prop:Partial agglomeration logit model}:}

\noindent When $t(x)=-\mu\ln(1-x)$, we have:
\[
t^{\prime}\left(h^{*}\right)-t^{\prime}\left(1-h^{*}\right)=\mu\left[\dfrac{1}{h(1-h)}\right].
\]
Using $\theta=1$ in $\tfrac{d\Delta u}{dh}$ given by (\ref{eq:stability of partial agglomeration}),
and rearranging, partial agglomeration is stable if:{\small{} 
\begin{equation}
-\frac{w^{-\sigma-1}\left[w^{2\sigma}-(w+1)\phi w^{\sigma}+w\right]^{2}\Gamma}{(\sigma-1)\left(w^{\sigma}-\phi\right)\left(1-\phi w^{\sigma}\right)\left[(\sigma-1)\phi+(\sigma-1)\phi w^{2\sigma}+\left(-2\sigma+\phi^{2}+1\right)w^{\sigma}\right]}<0,\label{eq:stab partial agg logit 1}
\end{equation}
}where 
\begin{align*}
\Gamma & =\phi\left[\mu(\sigma-1)^{2}-2\sigma+1\right]+\phi\left[\mu(\sigma-1)^{2}-2\sigma+1\right]w^{2\sigma}-\\
 & -w^{\sigma}\left\{ \phi^{2}\left[-(\mu+2)\sigma+\mu+1\right]+(2\sigma-1)\left[\mu(\sigma-1)-1\right]\right\} .
\end{align*}
The numerator of (\ref{eq:stab partial agg logit 1}) except $\Gamma$
is positive, and so is $\left(w^{\sigma}-\phi\right)\left(1-\phi w^{\sigma}\right)$.
We have: 
\[
(\sigma-1)\phi+(\sigma-1)\phi w^{2\sigma}+\left(-2\sigma+\phi^{2}+1\right)w^{\sigma}<0.
\]
Therefore, (\ref{eq:stability of partial agglomeration}) holds if
and only if $\Gamma<0$, which gives: 
\[
\mu>\mu_{p}\equiv\frac{(2\sigma-1)(w^{\sigma}-\phi)(1-w^{\sigma}\phi)}{(\sigma-1)\left[\left(2\sigma-\phi^{2}-1\right)w^{\sigma}-(\sigma-1)\phi w^{2\sigma}-(\sigma-1)\phi\right]}.
\]
The condition for partial agglomeration in (\ref{eq:stab partial agg logit 1})
holds for any interior equilibrium including symmetric dispersion
$h^{*}=1/2$. At symmetric dispersion we have $w=1$ and the condition
above simplifies to: 
\[
\mu>\mu_{d}\equiv\frac{(2\sigma-1)(1-\phi)}{(\sigma-1)(2\sigma+\phi-1)}.
\]
The derivative of $\mu_{p}$ with respect to $X\equiv w^{\sigma}$
equals: 
\[
\dfrac{\partial\mu_{p}}{\partial X}=\frac{\sigma(2\sigma-1)\left(X^{2}-1\right)\phi\left(\phi^{2}-1\right)}{(\sigma-1)\left[(\sigma-1)\left(X^{2}+1\right)\phi-2\sigma X+X\phi^{2}+X\right]^{2}}<0.
\]
We can also see that $\mu_{p}$ approaches zero as $w^{\sigma}$ approaches
$\phi^{-1}$. Therefore, we conclude that $0<\mu_{p}<\mu_{d}$. We
can show that symmetric dispersion undergoes a pitchfork bifurcation
at $\mu=\mu_{d}$, following a procedure similar to that of the proof
of Lemma \ref{lem:A-curve-of}\textbf{ }(see Appendix B) by replacing
$\phi$ with $\mu$ as the bifurcation parameter. The corresponding
derivatives have the same signs when evaluated at $\mu=\mu_{d}$.
Furthermore, we have:
\[
\dfrac{\partial^{3}\Delta V}{\partial h^{3}}(\tfrac{1}{2};\mu_{d})=-\frac{64(1-\phi)\phi\left\{ \sigma\left[\phi(\phi+2)+5\right]-\phi^{2}-3\right\} }{(\sigma-1)(\phi+1)^{3}(2\sigma+\phi-1)}<0,
\]
which ensures that the pitchfork is supercritical and proves that
a branch of stable partial equilibria emerges in the direction of
$\mu<\mu_{d}$. Since the state space is one-dimensional, multiple
partial agglomeration equilibria require the interchange of stability
between consecutive equilibria, which cannot happen when $\mu\in(\mu_{p},\mu_{d})$.\footnote{This holds even if equilibria are irregular.}
We thus have two possibilities: (i) if $\mu\in(\mu_{p},u_{d})$, partial
agglomeration is the only stable equilibrium and is unique; and (ii)
if $\mu>\mu_{d},$ symmetric dispersion is the only stable equilibrium. 

Next, we know from the text that a partial agglomeration equilibrium
always exists if $\mu<\mu_{d}$. Assume now, by way of contradiction,
that $\mu\in(0,\mu_{p})$. Then both dispersion and partial agglomeration
are unstable. Since the state space is one-dimensional, there can
be only one unstable partial agglomeration, which would require agglomeration
to be stable. However, we know that agglomeration is always unstable,
which implies that partial agglomeration is always stable. Hence,
$\mu\notin(0,\mu_{p})$. Therefore, symmetric dispersion is the only
stable equilibrium if $\mu>\mu_{d}$; otherwise, partial agglomeration
is the only stable equilibrium.\hfill{}$\square$ 
\end{document}